\documentclass[11pt,a4paper]{article}
\usepackage[left=24mm,right=24mm,top=25mm,bottom=25mm,headheight=14pt,headsep=8mm]{geometry}
\usepackage{iftex}
\ifPDFTeX
  \usepackage[T1]{fontenc}
  \usepackage[utf8]{inputenc}
  \usepackage{lmodern}
\else
  \usepackage{fontspec}
\fi
\usepackage{amsmath,amssymb,amsthm,mathtools}
\usepackage{booktabs,tabularx,array,graphicx}
\usepackage{enumitem,fancyhdr,caption,titlesec,placeins}
\usepackage[protrusion=true,expansion=false]{microtype}
\usepackage[hidelinks,unicode]{hyperref}
\hypersetup{pdftitle={Capacity of a Mixed Multiple-Access Channel with Causal Observation},pdfauthor={Xianwei Meng}}
\setlist{nosep,leftmargin=2em}

\titleformat{\section}{\normalfont\large\bfseries}{\thesection.}{0.65em}{}
\titleformat{\subsection}{\normalfont\normalsize\bfseries}{\thesubsection.}{0.65em}{}
\titlespacing*{\section}{0pt}{2.2ex plus .5ex minus .2ex}{1ex plus .2ex}
\titlespacing*{\subsection}{0pt}{1.8ex plus .4ex minus .2ex}{.7ex plus .2ex}
\theoremstyle{plain}
\newtheorem{theorem}{Theorem}
\newtheorem{proposition}[theorem]{Proposition}
\theoremstyle{definition}
\newtheorem{definition}[theorem]{Definition}
\newcommand{\Prb}{\mathbb P}
\newcommand{\E}{\mathbb E}
\newcommand{\F}{\mathbb F}
\newcommand{\ind}{\mathbf 1}
\newcommand{\pg}{p_{\mathrm g}}
\newcommand{\pb}{p_{\mathrm b}}
\newcommand{\Cbar}{\overline p}
\newcommand{\Rop}{\mathcal R_{\mathrm{op}}}
\begin{document}
\thispagestyle{plain}
\begin{center}
{\fontsize{18}{22}\selectfont\bfseries
Capacity of a Mixed Multiple-Access Channel\\[2pt]
with Causal Observation\par}
\vspace{10pt}
{\normalsize Xianwei Meng\textsuperscript{*}\par}
\vspace{3pt}
{\small Hefei University of Technology\par}
\vspace{3pt}
{\small\textsuperscript{*}Corresponding author: \href{mailto:mxianwei@hfut.edu.cn}{mxianwei@hfut.edu.cn}\par}
\end{center}
\vspace{2pt}
\begin{abstract}
\noindent
We determine the capacity region of a mixed binary multiple-access channel whose receiver can choose how each observation is formed. Two ports reveal the sum modulo two of the transmitted symbols with retention probabilities $\pg>\pb$; their roles are interchanged by an unknown state that remains fixed during the block. With one observation per channel use, no transmitter feedback, and zero switching cost, the region is $R_1+R_2\le\pg$ for every fixed error threshold below one, and a strong converse holds. For error thresholds below one half, the largest open-loop sum rate is $(\pg+\pb)/2$, whereas a fixed port permits only $\pb$. A causal receiver attains the larger region by learning the state from erasure flags while using every channel use for data. A posterior port rule makes only finitely many incorrect selections almost surely, with a uniformly bounded expected number of mistakes. Exact random linear-code formulas and output-counting converses quantify the finite-block consequences. At blocklength 256 with $(\pg,\pb)=(0.9,0.4)$, independent numerical validation supports 214 message bits at a one-percent ensemble failure target, compared with 148 bits for the best open-loop allocation in the same code ensemble. The resulting 44.59\% increase is achieved with the same transmission and observation budgets. The analysis exhibits an explicit capacity gain from controlling the observation axis and separates that gain from information lost after a record has already been formed.
\end{abstract}
\begingroup\small
\par
\endgroup

\section{Introduction}

In a mixed channel, a random component is selected at the beginning of a block and remains in force throughout transmission. A high average quality need not imply a high reliable rate: a component of positive probability may remain unfavorable throughout the block. In a multiple-access channel, this obstruction acts on the individual rates and on their sum. Han and Yagi~\cite{HanYagi2026} study information-spectrum descriptions of mixed multiple-access channels with cost constraints and different forms of state information. Here we ask what changes when the receiver can use its past observations to choose the channel through which the next record is obtained.

The physical--observation dual-axis formulation~\cite{MengDual2026} distinguishes the physical state from the observation state in the formation of a record. In a deterministic model the record is $Y=h(X,S)$; in a stochastic model it is described by a kernel $K(\mathrm dy\mid x,s)$. Fixing $S$ specifies an ordinary channel. Allowing $S$ to depend causally on the receiver's history gives a family of channels induced by observation policies. The distinction is operational: the encoder chooses the transmitted signal, while the receiver chooses which admissible record to acquire. Both choices enter the probability law on which a coding theorem must be based.

Action-dependent acquisition of state information has a well-developed information-theoretic formulation, including the probing-capacity problem of Asnani, Permuter, and Weissman~\cite{Asnani2010}. We consider a receiver-controlled example in which the observation law, capacity, learning cost, and finite-block decoder can all be analyzed explicitly. Two users transmit binary symbols. At each channel use the receiver reads one of two ports, obtaining either their sum modulo two or an erasure. An unknown binary block state interchanges the good and bad ports. Fixed observation, predetermined scanning, pilot-assisted selection, and continual posterior selection are thus protocols on the same physical--observation kernel.

The main result is an exact separation between these observation classes. At any error threshold below one half, their sum capacities are $\pb$, $(\pg+\pb)/2$, and $\pg$ for fixed, optimally allocated open-loop, and causal observation, respectively. The causal region satisfies a strong converse for every error threshold below one. Its proof uses a receiver that observes the good port and can simulate any admissible causal policy by adding erasures. This simulation preserves the causal order and relies on the absence of transmitter feedback. Achievability follows from binary linear coding after a vanishing fraction of pilot observations.

Dedicated pilots are not necessary. An erasure flag identifies whether a record was retained without revealing its value, and hence supplies state information during data transmission. We prove that a posterior port rule has a uniformly bounded expected number of incorrect selections and eventually uses the good port almost surely. A finite-horizon recursion then characterizes optimal observation within the class of policies that depend on erasure histories but not on the realized code matrix. Exact rank formulas and numerical comparisons show how the asymptotic gain appears at finite blocklength. Throughout, the performance gains refer to this specified channel and its resource constraints.

\section{Channel Model and Admissible Observation}

\subsection{Physical state and record formation}

The senders hold independent uniform messages $J_1$ and $J_2$ and use deterministic encoders $U_i^n=u_i^n(J_i)$ over $\F_2$. The state $H\in\{0,1\}$ is equiprobable, independent of the messages, and constant over the block. Its distribution and the channel parameters are known, but its realization is initially unknown. The physical and observation states at time $t$ are
\begin{equation}
 X_t=(U_{1t},U_{2t},H),\qquad S_t\in\{0,1\}.
 \label{eq:axes}
\end{equation}
The receiver reads only the selected port. Its output alphabet is $\mathcal Y=\{0,1,\bot\}$, where $\bot$ is an identifiable erasure. For $0<\pb<\pg<1$, put
\begin{align}
 p_{hs}&=\pg\ind\{s=h\}+\pb\ind\{s\ne h\},\label{eq:p}\\
 K(y\mid(u_1,u_2,h),s)
 &=p_{hs}\ind\{y=u_1\oplus u_2\}
 +(1-p_{hs})\ind\{y=\bot\}.\label{eq:kernel}
\end{align}
All physical--observation pairs are admissible. Conditional on the state and current action, a fresh independent random variable generates each erasure. Thus a change of observation state changes the chance of obtaining a record even when the physical state is held fixed. Information is measured in bits; $\ln$ denotes the natural logarithm.

\subsection{Causal policies and the induced channel}

Neither sender knows $H$ or receives feedback about the outputs or selected ports. The receiver may use its entire local history:
\begin{equation}
 S_t=\pi_t(S^{t-1},Y^{t-1},V),
 \label{eq:causal}
\end{equation}
where $V$ is a receiver seed independent of the messages and $H$. Deterministic policies are included. For the code and policy $\Pi$, the induced law is
\begin{equation}
 \Prb_{\Pi}(h,s^n,y^n\mid j_1,j_2,v)
 =\frac12\prod_{t=1}^n
 \ind\{s_t=\pi_t(s^{t-1},y^{t-1},v)\}
 K(y_t\mid(u_{1t}(j_1),u_{2t}(j_2),h),s_t).
 \label{eq:strategy-law}
\end{equation}
The decoder observes $R_n=(V,S^n,Y^n)$. Writing $J=(J_1,J_2)$, its channel law is therefore
\begin{equation}
 P_{R_n\mid J}(\mathrm dv,s^n,y^n\mid j)
 =P_V(\mathrm dv)\sum_{h=0}^1
 \Prb_{\Pi}(h,s^n,y^n\mid j,v).
 \label{eq:record-law}
\end{equation}
The hidden state is marginalized, while the receiver's visible randomness and control history remain part of the record.

If a delivered record $Z_n$ is obtained from $R_n$ by message-independent processing, the chain rule gives
\begin{equation}
 I(J;R_n)=I(J;Z_n)+I(J;R_n\mid Z_n).
 \label{eq:master}
\end{equation}
This identity measures the information lost in processing an existing record. Observation control acts earlier: it changes the law~\eqref{eq:record-law} under which the record is formed. Since the action is a function of information already available at the receiver,
\begin{equation}
 I(J;S_t\mid S^{t-1},Y^{t-1},V)=0.
 \label{eq:no-action-message}
\end{equation}
Its benefit comes through future observations. No additional message-bearing control link is present.

\subsection{Resources and error criterion}

The resource convention $\kappa$ fixes a block of $n$ channel uses. Each sender emits one binary symbol per use, and the receiver acquires exactly one sample. A transmitted symbol and an observation each have unit cost; port switching has zero cost. The per-block costs are consequently $n$ symbols per sender and $n$ observations. Pilot symbols count toward the same budgets. Message set sizes are fixed before transmission, and there is no early termination. These costs measure channel uses and samples; no physical energy model is assumed.

\begin{definition}[Operational region]
For $0\le\varepsilon<1$, a rate pair is achievable if a sequence of codes satisfying $\kappa$ has average joint decoding error with upper limit at most $\varepsilon$ and message rates with lower limits at least $R_1,R_2$. The closure of achievable pairs is denoted by $\Rop(\varepsilon;\kappa)$. The error probability averages over the messages, block state, channel randomness, and receiver seed. The case $\varepsilon=0$ is the vanishing-error criterion.
\end{definition}

\section{Capacity and the Cost of Predetermined Observation}

\begin{theorem}[Causal-observation capacity and strong converse]
For the channel~\eqref{eq:kernel} with the information pattern and resources above,
\begin{equation}
 \Rop(\varepsilon;\kappa)
 =\{(R_1,R_2):R_1\ge0,\ R_2\ge0,\ R_1+R_2\le\pg\},
 \qquad 0\le\varepsilon<1.
 \label{eq:capacity}
\end{equation}
For any sequence of codes whose sum rate exceeds $\pg$ by a fixed positive amount, the average error tends to one.
\end{theorem}

\begin{proof}
For the converse, give the receiver $H$ and a binary erasure channel that retains $U_{1t}\oplus U_{2t}$ with probability $\pg$ at every time. This receiver can reproduce any original policy. It first computes $S_t$ from the previously simulated history. If the enhanced channel retains the bit, it keeps that bit with probability $p_{HS_t}/\pg$ and otherwise erases it. Recursion in $t$ yields exactly~\eqref{eq:strategy-law}. The senders receive no feedback, so the additional erasures do not change their transmitted sequences.

Let $M=M_1M_2$ and let $B_n\sim\operatorname{Bin}(n,\pg)$ count retained positions in the enhanced channel. Conditional on their locations and on $H$, $r$ binary outputs distinguish at most $2^r$ joint messages. Decoder randomization cannot increase this count. Hence every original code and policy satisfy
\begin{equation}
 1-P_e\le\E\min\{1,2^{B_n}/M\}.
 \label{eq:converse}
\end{equation}
If $\log_2M\ge n(\pg+\delta)$, then
\[
 1-P_e\le
 \Prb\{B_n\ge\log_2M-n\delta/2\}+2^{-n\delta/2}\longrightarrow0.
\]
The side information $H$ is independent of the messages and adds no message-bearing output.

For achievability, use $m_n=2\lceil\ln n\rceil$ known pilot symbols, observing each port $m_n/2$ times, and select the port with more retained pilots. Section~\ref{sec:pilots} shows that its error probability tends to zero. Since $m_n/n\to0$, the number of retained data positions divided by $n$ converges in probability to $\pg$. Give the users independent random binary linear encoders with a total of $k_n$ message bits. For $k_n/n\to R_1+R_2<\pg$, the number of retained rows exceeds $k_n$ by a linear margin with probability tending to one. The rank formula~\eqref{eq:rank} then makes the ensemble failure probability tend to zero. For each blocklength a deterministic pair of matrices achieves at most the ensemble average. The bits may be split between the users in any nonnegative proportions, including the boundary cases $k_1=0$ or $k_2=0$. Taking the closure gives the boundary.
\end{proof}

\begin{theorem}[Fixed and open-loop observation]
Let $\Cbar=(\pg+\pb)/2$. For $0\le\varepsilon<1/2$, the largest open-loop sum rate is $\Cbar$, achieved by equal allocation to the two ports; a fixed port has sum capacity $\pb$. For $1/2\le\varepsilon<1$, a fixed port already has sum capacity $\pg$. In each case the corresponding region consists of nonnegative rate pairs below the stated sum constraint.
\label{thm:open-loop}
\end{theorem}

\begin{proof}
If port 0 is used for a fraction approaching $\alpha$, the conditional retention fractions are
\[
 c_0(\alpha)=\alpha\pg+(1-\alpha)\pb,
 \qquad c_1(\alpha)=\alpha\pb+(1-\alpha)\pg.
\]
Their minimum is $\Cbar-|\alpha-1/2|(\pg-\pb)$ and is largest at $\alpha=1/2$.

The converse does not require the allocation fractions to converge. For any length-$n$ allocation, at least one state has expected retention count at most $n\Cbar$. If $\log_2M\ge n(\Cbar+\delta)$, the variance bound $\operatorname{Var}(N\mid H)\le n/4$ and conditional output counting give success probability at most
\[
 \frac{1}{n\delta^2}+2^{-n\delta/2}
\]
in that state. Its prior probability is one half, so $\liminf P_e\ge1/2$. The bound is uniform in the allocation and also holds after conditioning on a message-independent randomized schedule.

Equal allocation makes the retention fraction converge to $\Cbar$ in either state; random linear coding achieves every smaller sum rate. With a fixed port the two thresholds are $\pg$ and $\pb$. Below error one half, the bad state must be served. At error one half, reliable transmission in the good state suffices. The universal converse~\eqref{eq:converse} excludes rates above $\pg$. All boundary rates follow by closure.
\end{proof}

At error thresholds below one half, the increase over the best open-loop sum capacity is
\begin{equation}
 G_S=\pg-\Cbar=\frac{\pg-\pb}{2},
 \qquad \frac{G_S}{\Cbar}=\frac{\pg-\pb}{\pg+\pb}.
 \label{eq:gain}
\end{equation}
For $(\pg,\pb)=(0.9,0.4)$, this is a 38.46\% increase. The channel family and resource budgets are unchanged; only the permission to adapt observation differs. The gain vanishes when the ports have equal quality. It also vanishes as a capacity separation once an error threshold of one half allows the receiver to abandon one block state.

\subsection{The information spectrum of the record}

Take independent fair data symbols and use a fixed schedule, an open-loop schedule, or a data-period port chosen from known pilots. The pilots, actions, and erasure pattern are independent of the data. If $N_{\mathrm{obs}}$ counts retained data positions, the three multiple-access information densities are
\begin{equation}
 \imath(U_1^d;R_n\mid U_2^d)
 =\imath(U_2^d;R_n\mid U_1^d)
 =\imath(U_1^d,U_2^d;R_n)=N_{\mathrm{obs}}.
 \label{eq:spectrum-count}
\end{equation}
Each retained binary equation contributes a factor of two to the relevant likelihood ratio. This equality concerns the stated independent input distribution and message-independent observation pattern.

At a fixed port, $N_{\mathrm{obs}}/n$ converges to $\pg$ or $\pb$ according to the block state. Its mean tends to $\Cbar$, while its limit inferior in probability is $\pb$. Equal scanning instead gives the limit $\Cbar$ in both states. Causal pilot selection with $m_n\to\infty$ and $m_n/n\to0$ gives the limit $\pg$. These spectra agree with the operational results and explain why averaging the component information rates overestimates the fixed-port vanishing-error capacity.

\section{Pilot Selection and Finite-Block Coding}
\label{sec:pilots}

\subsection{Learning the good port}

Let $m=2\ell$ be the pilot length. Each port is read $\ell$ times, and the remaining $d=n-m$ uses carry data. If $C_0,C_1$ are the pilot retention counts, their likelihood ratio is
\begin{equation}
 \ln\frac{\Prb(C_0,C_1\mid H=0)}{\Prb(C_0,C_1\mid H=1)}
 =(C_0-C_1)\ln\frac{\pg(1-\pb)}{\pb(1-\pg)}.
 \label{eq:llr}
\end{equation}
The coefficient is positive. Equal priors therefore lead to the port with the larger count, with a fair decision at ties. For independent $A\sim\operatorname{Bin}(\ell,\pg)$ and $B\sim\operatorname{Bin}(\ell,\pb)$, the exact wrong-port probability is
\begin{equation}
 q_m=\Prb(A<B)+\tfrac12\Prb(A=B).
 \label{eq:q}
\end{equation}
A fixed tie decision gives the same state-averaged value. Applying Hoeffding's inequality to the sum of independent differences yields
\begin{equation}
 q_m\le\exp\!\left[-\frac{m(\pg-\pb)^2}{4}\right].
 \label{eq:hoeffding}
\end{equation}
Thus consistent selection requires a growing number of pilots, but their fraction of the block may vanish.

\begin{proposition}[Optimal selection after balanced pilots]
Fix the pilot allocation and require a single port throughout the data period. Suppose both choices use the same code and a port-independent decoder for which adding erasures cannot improve success. The maximum-posterior rule~\eqref{eq:llr} minimizes average block failure. Unique recovery by full-rank linear decoding satisfies these conditions.
\end{proposition}
\begin{proof}
Given the pilots, put $w=\Prb(H=0\mid C_0,C_1)$. Let $E_{\mathrm g}\le E_{\mathrm b}$ be the failure probabilities at the good and bad ports. The two choices have conditional failures $wE_{\mathrm g}+(1-w)E_{\mathrm b}$ and $(1-w)E_{\mathrm g}+wE_{\mathrm b}$. Their difference is $(2w-1)(E_{\mathrm g}-E_{\mathrm b})$. Its sign gives the stated rule, and~\eqref{eq:llr} reduces that rule to a comparison of counts.
\end{proof}

The dual-axis formulation makes the observation choice explicit; in this restricted decision problem, its optimum is the usual Bayesian choice. The capacity gain follows from allowing that choice to affect future records.

\subsection{Random linear codes and unique recovery}

User $i$ sends $k_i$ uniform message bits. With $k=k_1+k_2$, generate
\begin{equation}
 G_1\in\F_2^{d\times k_1},\qquad
 G_2\in\F_2^{d\times k_2},\qquad
 G=[G_1\ G_2],\qquad U_i^d=G_iJ_i,
\end{equation}
using independent fair entries. The matrices are fixed and public after code design; no fresh shared randomness or matrix transmission is required. For the retained position set $\mathcal A$, the receiver solves
\begin{equation}
 G_{\mathcal A}\begin{bmatrix}J_1\\J_2\end{bmatrix}=Y_{\mathcal A}.
\end{equation}
It outputs the unique solution when one exists and declares failure otherwise. We count an erasure declaration as a block failure. This criterion differs from a decoder that guesses among multiple solutions.

For fixed, open-loop, pilot-assisted, and erasure-history policies, $\mathcal A$ is independent of the messages. With $G$ fixed, failure depends only on $\mathcal A$, so the average and maximal message failure probabilities coincide; both still average over $H$. General policies that use received message values need not have this property.

\begin{proposition}[Exact ensemble failure]
Suppose $\mathcal A$ is independent of the random matrix generated at design time, and put $N=|\mathcal A|$. Conditional on $N=r$, the probability of full column rank is
\begin{equation}
 F(r,k)=
 \begin{cases}
 \displaystyle\prod_{j=0}^{k-1}(1-2^{j-r}),&r\ge k,\\
 0,&r<k,
 \end{cases}
 \qquad F(r,0)=1.
 \label{eq:rank}
\end{equation}
Consequently, for $f_r=\Prb(N=r)$,
\begin{equation}
 \overline P_{\mathrm f}=1-\sum_r f_r F(r,k).
 \label{eq:ensemble}
\end{equation}
\end{proposition}
\begin{proof}
Generate the $r\times k$ retained matrix by columns. After $j$ independent columns, their span contains $2^j$ vectors. The next uniform column lies outside it with probability $1-2^{j-r}$. Multiplication gives~\eqref{eq:rank}; averaging over $N$ gives~\eqref{eq:ensemble}.
\end{proof}

Writing $E(d,k,p)=1-\E F(B_{d,p},k)$ for $B_{d,p}\sim\operatorname{Bin}(d,p)$, the fixed-port, equally scanned, and pilot-assisted protocols have
\begin{align}
 \overline P_{\mathrm f}^{\mathrm{fix}}&=\tfrac12E(n,k,\pg)+\tfrac12E(n,k,\pb),\label{eq:fix}\\
 \overline P_{\mathrm f}^{\mathrm{scan}}&=1-\E F(B_{n/2,\pg}+B_{n/2,\pb},k),\label{eq:scan}\\
 \overline P_{\mathrm f}^{\mathrm{pilot}}&=(1-q_m)E(n-m,k,\pg)+q_mE(n-m,k,\pb),\label{eq:act}
\end{align}
where $n$ is even in~\eqref{eq:scan} and the two binomial variables are independent. A general open-loop allocation uses port 0 for $a$ positions and port 1 for $n-a$ positions. Its count law is the equal mixture of the two corresponding binomial convolutions. Ensemble performance depends only on $a$. For a fixed code matrix, changing the schedule order can change the retained rows; an equivalent operation must permute the code coordinates and schedule together.

Conditional output counting also gives, for any code carrying $k$ bits under a fixed, open-loop, or pilot-assisted observation protocol,
\begin{equation}
 P_e\ge1-\E\min\{1,2^{N-k}\},
 \label{eq:finite-lower}
\end{equation}
where $N$ counts retained data positions. The pilots and erasure pattern are message-independent, which permits the conditioning. Unlike~\eqref{eq:ensemble}, this converse is not restricted to random linear codes or unique-recovery decoding. The broader converse~\eqref{eq:converse} applies to every causal policy.

\subsection{Information retained on unsuccessful blocks}

For the message-independent policies above, the posterior given $G$ and the received equations is uniform on an affine solution space of dimension $k-\operatorname{rank}(G_{\mathcal A})$. Hence, for a fixed matrix,
\begin{equation}
 \begin{split}
 H(J_1,J_2\mid R_n,G)&=k-\E\operatorname{rank}(G_{\mathcal A}),\\
 I(J_1,J_2;R_n\mid G)&=\E\operatorname{rank}(G_{\mathcal A}).
 \end{split}
 \label{eq:rank-information}
\end{equation}
The expectation is over the actual observation process. Full rank determines whether the whole block is recovered, whereas expected rank also measures information retained on failed blocks.

For ensemble calculations, let $v_r(j)$ be the rank distribution of a fair $r\times k$ binary matrix. Starting from $v_0(0)=1$, with out-of-range terms zero,
\begin{equation}
 v_{r+1}(j)=v_r(j)2^{j-k}
 +v_r(j-1)(1-2^{j-1-k}).
 \label{eq:rank-recursion}
\end{equation}
An additional row belongs to a $j$-dimensional row space with probability $2^{j-k}$. The recursion gives $v_r(k)=F(r,k)$ and, after averaging over $N$,
\[
 \overline I=\sum_r f_r\sum_j jv_r(j),\qquad
 \overline H=k-\overline I.
\]

\section{Data-Aided Causal Observation}
\subsection{Causal port selection from erasure flags}

Let $E_t=\ind\{Y_t\ne\bot\}$ and let
\[
 N_t=\sum_{i=1}^t E_i
\]
count the received message equations. We consider causal port-selection rules that use the past ports and erasure flags, together with an independent local random seed $V$. The rule does not use the message bits or the entries of the randomly generated codebook. The channel parameters are known and satisfy $0<\pb<\pg<1$.

Write
\[
 w_t=\Prb(H=0\mid S^t,E^t,V),\qquad
 L_t=\ln\frac{w_t}{1-w_t},\qquad w_0=\frac12,
\]
and define the flag kernel by
\begin{equation}
 b_{hs}(e)=p_{hs}^{\,e}(1-p_{hs})^{1-e},\qquad e\in\{0,1\}.
 \label{eq:online-flag-kernel}
\end{equation}
At time $t+1$, the posterior rule selects port $0$ when $w_t>1/2$ and port $1$ when $w_t<1/2$; ties may be resolved deterministically or with independent randomization. The log odds obey
\begin{equation}
 L_t=L_{t-1}+(1-2S_t)
 \left[E_t\ln\frac{\pg}{\pb}
 +(1-E_t)\ln\frac{1-\pg}{1-\pb}\right].
 \label{eq:online-log-odds}
\end{equation}
Thus the flag supplies evidence about the hidden state, while each unerased value $U_{1t}\oplus U_{2t}$ supplies a message equation. Every slot can carry data, with one port read per slot.

Let $G$ denote the codebook generated before transmission and put $J=(J_1,J_2)$. Conditional on the independent seed $V=v$, the state, actions, and flags have joint law
\begin{equation}
 \Prb(h,s^n,e^n\mid G,J,V=v)
 =\frac12\prod_{t=1}^n
 \ind\{s_t=\pi_t(s^{t-1},e^{t-1},v)\}\,b_{hs_t}(e_t).
 \label{eq:online-law}
\end{equation}
The right-hand side does not depend on $G$ or $J$. In particular, the set of successful positions is independent of the codebook and messages, so the full-rank formula for random linear codes gives
\begin{equation}
 \overline P_{\mathrm f}^{\mathrm{online}}
 =1-\sum_{r=0}^{n}\Prb(N_n=r)F(r,k).
 \label{eq:online-rank-error}
\end{equation}
This is the exact codebook-averaged probability of failure to recover the message pair uniquely. In general, $N_n$ is not binomial, because the port choices depend on the preceding flags. Numerical evaluation of its distribution by sampling introduces Monte Carlo error into the evaluation of~\eqref{eq:online-rank-error}.

\subsection{Bounded learning loss and capacity attainment}

In the swapped-port model, both ports have the same affinity between their flag distributions under the two hidden states:
\begin{equation}
 \gamma=\sqrt{\pg\pb}+\sqrt{(1-\pg)(1-\pb)}<1.
 \label{eq:online-affinity}
\end{equation}
Every observation therefore contributes to state discrimination, whichever port is selected. This symmetry yields a learning loss whose expectation remains bounded as the block length grows.

\begin{theorem}[Learning loss and capacity attainment]
Suppose that the channel parameters are known, the hidden state is constant throughout the block, the erasure flags are observed, and switching has no cost. Let every slot carry data, and select the ports by the posterior rule associated with~\eqref{eq:online-log-odds}. If $q_t=\Prb(S_t\ne H)$, then
\begin{align}
 q_t&\le\frac12\gamma^{t-1},\label{eq:online-wrong-bound}\\
 \E\sum_{t=1}^{n}\ind\{S_t\ne H\}
 &\le\frac{1-\gamma^n}{2(1-\gamma)},\label{eq:online-mistakes}\\
 \E N_n&\ge n\pg-
 \frac{(\pg-\pb)(1-\gamma^n)}{2(1-\gamma)}.
 \label{eq:online-record-mean}
\end{align}
Moreover, $N_n/n\to\pg$ almost surely. Combined with joint random linear coding, the rule achieves the capacity region in~\eqref{eq:capacity}.
\end{theorem}

\begin{proof}
First fix the independent random seed. Let $P_h^{(t)}$ be the law of the action and flag history through time $t$ under $H=h$, and let
\[
 A_t=\sum_{g_t}\sqrt{P_0^{(t)}(g_t)P_1^{(t)}(g_t)},\qquad A_0=1.
\]
At a given history, the causal rule selects the same action under either hypothesis. Since
$\sum_e\sqrt{b_{0s}(e)b_{1s}(e)}=\gamma$ for each port $s$, summing over the histories gives
\begin{equation}
 A_{t+1}
 =\sum_{g_t}\sqrt{P_0^{(t)}(g_t)P_1^{(t)}(g_t)}\,\gamma
 =\gamma A_t=\gamma^{t+1}.
 \label{eq:online-affinity-recursion}
\end{equation}
The same recursion holds for a randomized action kernel: its common conditional probabilities sum to one under the two hypotheses. Averaging over the independent seed preserves the bound below.

The posterior port choice at time $t+1$ is the maximum a posteriori estimate of the hidden state from the first $t$ observations. Under equal priors,
\[
 q_{t+1}=\frac12\sum_{g_t}
 \min\{P_0^{(t)}(g_t),P_1^{(t)}(g_t)\}
 \le\frac12 A_t=\frac12\gamma^t.
\]
Summation gives~\eqref{eq:online-mistakes}. Also,
\[
 \E(E_t\mid H,S^t,E^{t-1},V)
 =\pg-(\pg-\pb)\ind\{S_t\ne H\}.
\]
Taking expectations and summing over $t$ proves~\eqref{eq:online-record-mean}.

Because $\sum_t q_t<\infty$, the first Borel--Cantelli lemma implies that the inferior port is selected only finitely often, almost surely. On the other hand,
\[
 E_t-\bigl[\pg-(\pg-\pb)\ind\{S_t\ne H\}\bigr]
\]
is a bounded martingale difference. The strong law for such differences shows that its partial sums, divided by $n$, tend almost surely to zero. Hence $N_n/n\to\pg$.

If $k_n/n\to R_1+R_2<\pg$, choose $\delta>0$ so that
$\Prb(N_n<k_n+n\delta)\to0$. The product formula~\eqref{eq:rank} gives
\[
 1-F(r,k)\le\sum_{j=0}^{k-1}2^{j-r}<2^{k-r}
 \quad(r\ge k).
\]
Consequently,
\[
 \overline P_{\mathrm f}^{\mathrm{online}}
 \le\Prb(N_n<k_n+n\delta)+2^{-n\delta}\longrightarrow0.
\]
Here the independence of the successful positions and the random coding matrix is essential. Since the codebook-averaged failure probability tends to zero, a deterministic codebook sequence with vanishing failure probability exists. The strong converse for causal strategies, together with closure of the achievable region, establishes~\eqref{eq:capacity}.
\end{proof}

For $\pg=0.9$ and $\pb=0.4$, one has $\gamma\approx0.844949$. The theorem bounds the expected number of inferior-port selections by $3.224745$, uniformly in the block length. The corresponding expected loss of received equations, relative to knowing the better port in advance, is at most $1.612373$. These are analytical upper bounds. Learning during data transmission replaces a dedicated training interval by a bounded expected loss of equations.

A finite-block bound follows from the same argument. For integers $0\le T\le n$ and $\lambda\ge0$,
\begin{equation}
 \overline P_{\mathrm f}^{\mathrm{online}}
 \le\frac{\gamma^T}{2(1-\gamma)}
 +\Prb\!\left\{\operatorname{Bin}(n,\pg)<k+\lambda+T\right\}
 +2^{-\lambda}.
 \label{eq:online-finite-bound}
\end{equation}
To prove this, use a common sequence of independent uniform random variables to couple the observed success count with the count $B_n$ obtained by always selecting the better port. If no inferior port is selected after time $T$, then $N_n\ge B_n-T$. The probability of at least one such selection after time $T$ is bounded by
$\sum_{t=T+1}^n q_t\le\gamma^T/[2(1-\gamma)]$. Combining this estimate with the full-rank failure bound proves~\eqref{eq:online-finite-bound}. Minimizing over $T$ and $\lambda$ improves the bound, although it can remain substantially larger than an evaluation based on the full distribution of $N_n$.

\subsection{Finite-block optimization of the observation rule}

The posterior rule maximizes the probability of success in the next slot. This one-step property does not establish its optimality for block decoding: an observation also changes the information available for subsequent choices. Dynamic programming accounts for both effects.

After $t$ observations, let $w$ be the posterior probability of state $0$ and let $r$ be the number of received equations. For action $s$, set
\begin{align}
 a_s(w)&=wp_{0s}+(1-w)p_{1s},\\
 T_{s,1}(w)&=\frac{wp_{0s}}{a_s(w)},&
 T_{s,0}(w)&=\frac{w(1-p_{0s})}{1-a_s(w)}.
 \label{eq:online-belief-update}
\end{align}
Within the policy class considered here, let $W_t(w,r)$ be the largest codebook-averaged probability of unique recovery achievable with the remaining $n-t$ observations. Then
\begin{align}
 W_n(w,r)&=F(r,k),\label{eq:online-bellman-terminal}\\
 W_t(w,r)&=\max_{s\in\{0,1\}}
 \left\{a_s(w)W_{t+1}\bigl(T_{s,1}(w),r+1\bigr)
 +[1-a_s(w)]W_{t+1}\bigl(T_{s,0}(w),r\bigr)\right\}.
 \label{eq:online-bellman}
\end{align}
The policies in~\eqref{eq:online-bellman} know $n,k,\pg,\pb$ and use only the action and flag history. Under this restriction, the successful positions remain independent of the random coding matrix, and the conditional, codebook-averaged probability of recovery is $F(r,k)$. The pair $(w,r)$ is therefore a sufficient state. A policy allowed to inspect the row space of a particular codebook would require an enlarged state and a different terminal reward.

Equation~\eqref{eq:online-bellman} gives the exact finite-block optimum within this class. The blocklength-256 results below use the posterior rule; the Bellman optimum is evaluated only for the shorter blocks stated there. Switching costs, an evolving hidden state, or imperfectly observed erasure flags would change the posterior dynamics or the admissible policies and require a new analysis.

\section{Finite-Block Performance}
\subsection{Reliable payload and observation cost}

We use $(\pg,\pb)=(0.9,0.4)$, equal state probabilities, and a target block failure probability of one percent. The users carry the same number of bits, so $k$ is even. Every open-loop allocation $a=0,\ldots,n$ is evaluated. The pilot-assisted rate comparison searches even $m\le\min(96,n-2)$; comparisons at a fixed payload search all feasible even $m\le n-k$. Rates are normalized by the entire block, including pilots.

Table~\ref{tab:main} gives the principal comparison at $n=256$. The fixed, open-loop, and pilot-assisted values follow from finite sums, without Monte Carlo sampling. For the online policy, the distribution of $N_n$ is estimated from independent validation blocks and averaged analytically over the random code matrix using~\eqref{eq:online-rank-error}. Thus the first three columns are exact ensemble calculations up to numerical precision, while the last column contains sampling estimates.

\begin{table}[htbp]
\centering\small
\caption{Performance at $n=256$ and a one-percent ensemble failure target. The online column is estimated from $10^6$ independent validation blocks. All protocols use 512 transmitted symbols and 256 observations.}
\label{tab:main}
\setlength{\tabcolsep}{5pt}
\begin{tabular}{lrrrr}
\toprule
Quantity & Fixed & Open-loop & Pilots & Online\\
\midrule
Total message bits $k$ & 84 & 148 & 196 & 214\\
Net sum rate $k/n$ & 0.328125 & 0.578125 & 0.765625 & 0.835938\\
Block failure probability & 0.007569 & 0.005436 & 0.007362 & 0.004871\\
Goodput (bits/use) & 0.325641 & 0.574983 & 0.759988 & 0.831866\\
Pilot uses & 0 & 0 & 20 & 0\\
Mean retained data equations & 166.400 & 166.400 & 211.769 & 229.609\\
1\% quantile of retained equations & 86 & 151 & 200 & 217\\
Residual message entropy (bits) & 0.016257 & 0.008884 & 0.545659 & 0.007046\\
Observations per delivered bit & 3.0709 & 1.7392 & 1.3158 & 1.2021\\
Mean port switches & 0.0000 & 1.0000 & 1.5000 & 1.0563\\
\bottomrule
\end{tabular}
\end{table}

The pilot scheme uses 20 observations for training and carries 196 bits, a 32.43\% increase over the optimized open-loop ensemble's 148 bits. Online selection carries 214 bits, increasing this benchmark by 44.59\% and the pilot payload by 9.18\%. These are finite-block comparisons within the stated ensemble and unique-recovery criterion. The asymptotic capacity increase in~\eqref{eq:gain} is a different quantity.

For context, output counting gives an upper bound of 152 bits for any open-loop code at the same one-percent error target, under the even-$k$ convention. The corresponding bounds are 200 bits for the balanced-pilot, single-data-port class and 218 bits for all causal protocols. The 196-bit ensemble calculation therefore implies the existence of a deterministic pilot-assisted code exceeding the arbitrary-code open-loop upper bound by 28.95\%. The online result is four bits below the causal converse; the exact finite-block optimum is not determined.

We define goodput as $T=(k/n)(1-\overline P_{\mathrm f})$, assigning zero delivered bits to a failed block. Online goodput is estimated at $0.831866$ bits per use, 44.68\% above the open-loop value. Since the transmission and observation counts are fixed, bits delivered per transmitted symbol equal $T/2$, and bits delivered per observation equal $T$. The users have equal nominal rates and equal goodput; their Jain index is identically one. All blocks last 256 channel uses. Switching counts are reported, but the model assigns them zero cost and imposes no hard switching limit on the online policy.

\begin{figure}[htbp]
\centering
\includegraphics[width=\linewidth]{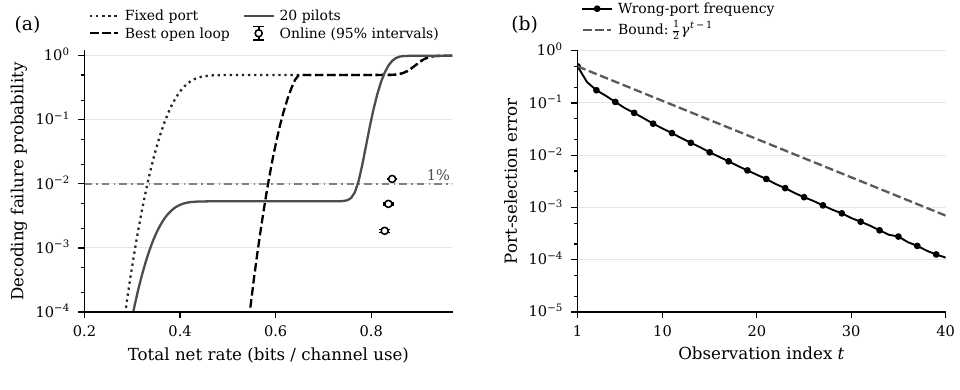}
\caption{Finite-block reliability and state learning for $(\pg,\pb)=(0.9,0.4)$. (a) Exact ensemble failure curves for fixed observation, payload-optimized open-loop allocation, and a fixed 20-pilot scheme at $n=256$. Online points and simultaneous 95\% empirical Bernstein intervals come from independent validation; the dotted horizontal line marks one percent. (b) Observed wrong-port frequency during online selection and the analytical bound $\gamma^{t-1}/2$.}
\label{fig:performance}
\end{figure}

\subsection{Independent validation and fixed-code decoding}

For the online protocol, $200{,}000$ blocks were used to choose the candidate payloads 212, 214, and 216 bits. These three candidates were fixed before a separate set of $L=10^6$ validation blocks was examined. For a candidate $k$, each block contributes the bounded quantity
\[
 Z_i=1-F(N_n^{(i)},k)\in[0,1].
\]
This averages out code-matrix randomness exactly; the sampling variation comes from the state and causal erasure process. Let $\widehat e$ and $\widehat v$ be the sample mean and unbiased sample variance. The empirical Bernstein inequality of Maurer and Pontil~\cite{MaurerPontil2009}, applied to both tails of each of $J=3$ candidates, gives the radius
\begin{equation}
 b_L=\sqrt{\frac{2\widehat v\ln(4J/\delta)}{L}}
 +\frac{7\ln(4J/\delta)}{3(L-1)},\qquad \delta=0.05.
 \label{eq:validation-radius}
\end{equation}
Intersecting $[\widehat e-b_L,\widehat e+b_L]$ with $[0,1]$ yields simultaneous coverage of at least 95\% for the three ensemble failure probabilities.

The resulting intervals, in percent, are $[0.1749,0.1956]$ for 212 bits, $[0.4702,0.5040]$ for 214 bits, and $[1.1677,1.2220]$ for 216 bits. The largest validated candidate meeting the target is therefore 214 bits. Other entries in the online column of Table~\ref{tab:main} are descriptive estimates and are not assigned these simultaneous intervals.

A separate experiment fixes a public code matrix, generates actual messages, and performs binary elimination. At 214 bits, the online receiver failed in 16 of 4000 transmissions: a frequency of 0.400\%, with nominal 95\% Wilson interval $[0.246\%,0.649\%]$. Every uniquely recovered message was correct. The pilot-assisted fixed code at 196 bits failed in 40 of 4000 transmissions, with interval $[0.735\%,1.359\%]$. The latter trial does not establish that this particular matrix meets the one-percent target, despite the lower exact ensemble average. Fixed-code trials and ensemble validation estimate different probabilities.

\subsection{Equal payloads and finite-horizon optimization}

Maximizing payload at a prescribed reliability and minimizing failure for a prescribed payload need not select the same observation schedule. At $k=196$, the best open-loop choice is a fixed port, which succeeds essentially only in the good state. At the one-percent payload optimum $k=148$, equal scanning is best. Table~\ref{tab:common} also shows that block failure and residual entropy can order the protocols differently.

\begin{table}[htbp]
\centering\small
\caption{Exact ensemble comparisons at a common payload, $n=256$. The open-loop allocation and the optimized pilot length are chosen separately for each $k$.}
\label{tab:common}
\setlength{\tabcolsep}{6pt}
\begin{tabular}{rlrrr}
\toprule
$k$ & Observation protocol & Block failure & Goodput & Residual entropy\\
&&& (bits/use)& (bits)\\
\midrule
148 & Optimized open-loop & 0.005436 & 0.574983 & 0.008884\\
148 & 20 pilots & 0.005346 & 0.575035 & 0.286526\\
148 & 62 pilots & $4.216\times10^{-6}$ & 0.578123 & 0.000194\\
\addlinespace
196 & Optimized open-loop & 0.500000 & 0.382812 & 46.800000\\
196 & 20 pilots (optimized) & 0.007362 & 0.759988 & 0.545659\\
\bottomrule
\end{tabular}
\end{table}

At 148 bits, 20 pilots slightly reduce block failure relative to scanning but leave more residual entropy. A mistaken port decision may lose many message equations, while an unsuccessful scanned block is usually closer to full rank. At the lighter payload of 84 bits, scanning has failure about $4.45\times10^{-21}$ and is substantially better than the fixed 20-pilot scheme. Increasing the mean number of records therefore does not ensure better reliability for every payload. The best pilot length depends on the decoding objective.

The same training tradeoff appears with blocklength. At lengths 64, 128, 256, and 512, the optimized open-loop ensemble supports rates $0.468750$, $0.546875$, $0.578125$, and $0.601562$ bits per use at one-percent failure. Within the stated pilot search, the corresponding rates are $0.468750$, $0.671875$, $0.765625$, and $0.824219$. The training cost erases the rate advantage at length 64 but becomes less important in longer blocks. The full parameter sweeps, including unequal known state priors and smaller port contrasts, are retained with the numerical data.

Finally, the Bellman recursion~\eqref{eq:online-bellman} was evaluated for all positive even payloads at $n=8,12,16,20$, giving 28 cases. Histories were represented by each port's visit and retention counts, four integers that determine the posterior. An independent forward probability calculation checked the resulting policies. Representative values appear in Table~\ref{tab:dp}.

\begin{table}[htbp]
\centering\small
\caption{Short-block ensemble failure probabilities. Open-loop allocation and pilot length are optimized at each payload. The Bellman policy is optimal within the erasure-history, codebook-independent policy class.}
\label{tab:dp}
\begin{tabular}{rrrrrr}
\toprule
$(n,k)$ & Fixed & Open-loop & Pilots & Posterior & Bellman\\
\midrule
$(8,4)$ & 0.466104 & 0.412563 & 0.466104 & 0.204104 & 0.204104\\
$(12,8)$ & 0.574353 & 0.574353 & 0.574353 & 0.286315 & 0.286315\\
$(16,10)$ & 0.521769 & 0.521769 & 0.396858 & 0.134119 & 0.134119\\
$(20,12)$ & 0.500774 & 0.489470 & 0.307519 & 0.059298 & 0.059298\\
\bottomrule
\end{tabular}
\end{table}

Across the 28 cases, the initial success values for the Bellman and posterior rules differ by at most $5.6\times10^{-17}$, and no strictly better nongreedy action was found within numerical precision. At $n=20$, each payload requires 10,626 cached states. These calculations show agreement for the tested instances; the blocklength-256 results use the posterior rule and make no claim of finite-horizon global optimality.

\section{Conclusion}

For the mixed binary multiple-access channel considered here, reliable transmission is limited by both the physical mixture and the receiver's admissible observation policy. A fixed port leaves the bad state as a persistent obstruction. Equal scanning balances the states. Causal observation learns which port is favorable and attains the good-port region, with the same number of transmitted symbols and acquired samples. The exact sum capacities $\pb$, $(\pg+\pb)/2$, and $\pg$ quantify these three possibilities at error thresholds below one half.

The learning cost need not occupy separate transmission uses. Erasure flags provide state evidence while retained bits provide message equations. The posterior rule incurs only a bounded expected number of incorrect port selections and is capacity-achieving without pilots. At finite blocklength, rank-based achievability and output-counting converses distinguish the performance of a particular code ensemble from the limit over all codes. The numerical results place the online scheme close to that converse in the stated example, while leaving its exact finite-block optimum open.

The physical--observation distinction is consequential because control precedes record formation. No subsequent processing can recover an erased sample, but a causal observation choice can improve the probability that later samples are retained. For this channel, the observation constraint is therefore part of the capacity problem itself.

Two extensions would alter the communication problem substantially: a state that evolves within the block, and observation actions with switching or quality-dependent costs. Phase-sensitive receivers introduce a further issue, the comparison of records obtained under different observation states; the phase-transport formulation~\cite{MengPhase2026} provides a geometric framework for that question. Extending the capacity calculation would require coupling those record comparisons to an explicit communication kernel and its observation costs.

\paragraph{Data and code availability.}
The accompanying source package contains the simulation programs, fixed code matrices, parameter files, seeds, and numerical results. Compilation instructions and detailed reproduction commands are provided in its README.

\end{document}